\documentclass[11pt]{article}

\usepackage[utf8]{inputenc} 
\usepackage[T1]{fontenc}    
\usepackage{hyperref}       
\usepackage{url}            
\usepackage{booktabs}       
\usepackage{amsfonts}       
\usepackage{nicefrac}       
\usepackage{microtype}      

\usepackage{multirow}
\usepackage{amstext}
\usepackage{amsmath}
\usepackage{amsthm}
\usepackage{amssymb}
\usepackage{bm}
\usepackage{wrapfig}
\usepackage{amsfonts, amsmath}
\usepackage[utf8]{inputenc}
\usepackage{graphicx}
\usepackage{bbm, comment}
\usepackage{subfigure}
\usepackage{color}
\usepackage{float}
\usepackage{wrapfig}
\usepackage{enumitem}
\usepackage{url}
\usepackage{MnSymbol,wasysym,marvosym,ifsym}

\def\lc{\left\lfloor}   
\def\rc{\right\rfloor}

\usepackage[ruled]{algorithm2e} 

\makeatletter
\def\BState{\State\hskip-\ALG@thistlm}
\makeatother

\usepackage{scalerel,stackengine}
\stackMath
\newcommand\reallywidehat[1]{%
\savestack{\tmpbox}{\stretchto{%
  \scaleto{%
    \scalerel*[\widthof{\ensuremath{#1}}]{\kern-.6pt\bigwedge\kern-.6pt}%
    {\rule[-\textheight/2]{1ex}{\textheight}}
  }{\textheight}%
}{0.5ex}}%
\stackon[1pt]{#1}{\tmpbox}%
}
\usepackage{tabulary}
\usepackage{booktabs}

\usepackage[top=1.in, bottom=1.in, left=1.in, right=1.in]{geometry}
\usepackage[numbers]{natbib}

\newcommand{\full}[0]{\text{permutation}}

\newtheorem*{theorem*}{Theorem}
\newtheorem{theorem}{Theorem}[section]
\newtheorem{lemma}[theorem]{Lemma}
\newtheorem{fact}[theorem]{Fact}
\newtheorem{claim}[theorem]{Claim}
\newtheorem{assumption}[theorem]{Assumption}

\newtheorem{example}[theorem]{Example}

\newtheorem{definition}[theorem]{Definition}
\newtheorem{observation}[theorem]{Observation}

\newcommand{\dmi}{\textsc{DMI}}
\newcommand{\E}{\mathrm{E}}
\newcommand{\var}{\mathrm{Var}}
\newcommand{\cov}{\mathrm{Cov}}
\newcommand{\sign}{\mathrm{sgn}}
\newcommand{\per}{\mathrm{per}}

\begin{document}
\title{Dominantly Truthful Multi-task Peer Prediction with a Constant Number of Tasks}
\author{Yuqing Kong \\
The Center on Frontiers of Computing Studies,\\
Computer Science Dept., Peking University \\
\texttt{yuqing.kong@pku.edu.cn} \\}
\date{}
\maketitle
\begin{abstract}

In the setting where participants are asked multiple similar \emph{possibly subjective}  multi-choice questions (e.g. Do you like Panda Express? Y/N; do you like Chick-fil-A? Y/N), a series of peer prediction mechanisms are designed to incentivize honest reports and some of them achieve \emph{dominantly truthfulness}: truth-telling is a dominant strategy and strictly dominate other ``non-permutation strategy'' with some mild conditions. However, a major issue hinders the practical usage of those mechanisms: they require the participants to perform an infinite number of tasks. When the participants perform a finite number of tasks, these mechanisms only achieve approximated dominant truthfulness. The existence of a dominantly truthful multi-task peer prediction mechanism that only requires a finite number of tasks remains to be an open question that may have a negative result, even with full prior knowledge. 

This paper answers this open question by proposing a new mechanism, Determinant based Mutual Information Mechanism (\dmi-Mechanism), that is dominantly truthful \emph{when the number of tasks is $\geq 2C$}. $C$ is the number of choices for each question ($C=2$ for binary-choice questions). \dmi-Mechanism also pays truth-telling higher than any strategy profile and strictly higher than uninformative strategy profiles (informed truthfulness). In addition to the truthfulness properties, \dmi-Mechanism is also easy to implement since it does not require any prior knowledge (detail-free) and only requires $\geq 2$ participants. The core of \dmi-Mechanism is a novel information measure, Determinant based Mutual Information (DMI). DMI generalizes Shannon's mutual information and the square of DMI has a simple unbiased estimator. In addition to incentivizing honest reports, \dmi-Mechanism can also be transferred into an information evaluation rule that identifies high-quality information without verification when there are $\geq 3$ participants.

\emph{To the best of our knowledge, \dmi-Mechanism is both the first detail-free informed-truthful mechanism and the first dominantly truthful mechanism that works for a finite number of tasks, not to say a small constant number of tasks.}
\end{abstract}

\section{Introduction}

Eliciting information is common in today's world (e.g. restaurants rating, movie rating) and the elicited information can affect people's life after being fed into algorithms (e.g. recommendation algorithms). However, the overwhelming number of elicitation requests may lead to unrepresentative feedbacks. The naive flat-payment reward, paying every participator one dollar regardless of her feedback, may distort people's incentives: some people may answer a large number of questions simply for the reward without making any attempt to answer accurately. In this case, two questions arise: 

\begin{itemize}
    \item Information Elicitation: how to encourage high-quality, honest information?
    \item Information Evaluation: how to identify high-quality information?
\end{itemize}

These two questions face the same crucial challenge: \emph{the information may not be verifiable, as the information can be subjective or the gold-standard/ground truth is hard to access}. With this challenge, spot-checking, i.e. verifying the answer with some probability, is impossible to implement. The traditional majority vote rule discourages the honest feedback and fails to identify the high-quality information \emph{from the minority}. 

To this end, peer prediction mechanisms arise to address this challenge. The general idea of peer prediction is to reward each person a clever similarity measure between her report and her peer's report such that honest strategy, i.e., truth-telling, is preferred in some solution concepts (e.g. truth-telling is a strict Nash equilibrium which means for each participator, given other people tell the truth, it's strictly better for her to tell the truth as well, even she belongs to the minority). Moreover, every peer prediction mechanism can be transferred into an information evaluation rule that does not need ground truth: scoring each participant's information based on her payment, since intuitively, in expectation, high-quality information should have high payment to be encouraged. 

The multi-task setting, where participants are asked several \emph{similar multi-choice questions} (e.g. Do you like Panda Express? Y/N; do you like Chick-fil-A? Y/N), is one of the common settings for peer prediction mechanisms. Although the multi-task setting asks the participants to perform multiple tasks, it has several advantages over the single-task setting. For example, unlike the single-task mechanisms, the multi-task mechanisms can be \emph{minimal} in the sense that the participants do not need to report their forecasts for other people. Moreover, in the multi-task setting, the single-task setting's \emph{common prior} assumption is unnecessary: the participants can share different prior knowledge for the tasks and be heterogeneous. We want the multi-task mechanisms to be

\begin{itemize}
\item \emph{dominantly truthful}: truth-telling is a dominant strategy and strictly dominates other ``non-permutation strategy''\footnote{Permutation strategy means always reporting a permuted version of the answer, e.g. say ``like'' when the honest answer is ``dislike'' while say ``dislike'' when the honest answer is ``like''. } with some mild conditions;
\item \emph{informed-truthful}: truth-telling is an equilibrium and moreover, truth-telling is the best\footnote{The ``best'' strategy profile has the highest amount of expected payment for every participant.} strategy profile and strictly better than uninformative strategy profiles, where everyone's report is independent of her honest answer;
\item \emph{detail-free}: the implementation of the mechanism does not require any prior knowledge;
\item \emph{practical}: the mechanism works for \emph{a small number of tasks} and a small number of participants. 
\end{itemize}

In the multi-task setting mechanisms, designing truthful mechanisms that works for a small number of participants properties is not hard. However, it's quite difficult to design truthful mechanisms with a small number of tasks and the current state-of-the-art requires participants to perform an infinite number of tasks to achieve informed/dominant truthfulness, which is certainly not practical. In detail, for the informed truthfulness, to reduce the number of tasks, either some prior knowledge is needed (not detail-free) or the truthfulness goal is replaced by its approximated version\footnote{$(\epsilon,\delta)$-informed truthfulness: truth-telling is at least $\epsilon$ better than other strategy profiles, with $1-\delta$ probability.}. For the dominant truthfulness, the current state-of-the-art requires an infinite number of tasks even with full prior knowledge. The mechanisms that work for a finite number of tasks only have an approximated dominant truthfulness\footnote{$(\epsilon,\delta)$-dominant truthfulness: truth-telling is at least $\epsilon$ better than other strategies, with $1-\delta$ probability.}. 

The key issue faced by the previous works is the need to learn the information structures among the participants. For example, when the participants have similar tastes, the mechanism should pay them for agreements. When the participants have opposite tastes, the mechanism should pay them for disagreements. When the participants' tastes have complicated information structures, the mechanism should pay them for ``clever'' (dis)agreements. To learn the information structure, the mechanism needs sufficient amount of tasks. To this end, the existence of a detail-free dominantly/informed-truthful mechanism even with a finite number of tasks remains to be an open question that may have a negative result. This work gives a positive answer to this open question by proposing a new mechanism. Unlike the previous mechanisms, this new mechanism does not have the learning issue since its implementation is independent of the information structure knowledge. 

\paragraph{Main contribution} This work considers both the information elicitation and the information evaluation questions, focuses on the multi-task setting and provides

\begin{itemize}
    \item Information Elicitation: a dominantly truthful, informed-truthful, detail-free and practical multi-task mechanism, Determinant based Mutual Information Mechanism (\dmi-Mechanism), which works for $\geq 2$ participants and $\geq 2C$ tasks (Theorem~\ref{thm:main}); 
    
    \item Information Evaluation: a multi-task information evaluation rule which is based on \dmi-Mechanism. This scoring rule assigns high-quality information a higher score in expectation (Theorem~\ref{thm:main2}). A concentration bound analysis for the scores is also provided (Section~\ref{sec:cen}). 
\end{itemize}

To the best of our knowledge, \dmi-Mechanism is both the first detail-free informed-truthful and the first detail-free dominantly truthful mechanism that works for a finite number of tasks, not to say a constant number of tasks.

\paragraph{Main technical contribution} The core of \dmi-Mechanism is a novel matrix determinant based information measure, DMI (Section~\ref{sec:dmi}). Like Shannon's mutual information (MI), DMI is non-negative, symmetric and information-monotone. DMI also has two additional desirable properties that MI does not have: DMI's unbiased estimator is easy to be constructed due to the polynomial format\footnote{Without the information structure knowledge, the information measures used in previous works \cite{Kong:2019:ITF:3309879.3296670,2016arXiv160303151S,liuchen} do not have a polynomial format. This is the main technical reason that they need an infinite number of tasks.} of the determinant and DMI has a special relative invariance property due to the multiplicative property of the determinant. The first property allows \dmi-Mechanism to be dominantly truthful with a constant number of tasks and the second property allows \dmi-Mechanism based scoring rule to assign a higher score to high-quality information in expectation.

\subsection{Roadmap}

This paper first focuses on the information elicitation. Section~\ref{sec:prelim1} introduces the formal setting for information elicitation without verification in the multi-task setting. Section~\ref{sec:dmi} introduces the definition and the properties of the main technical ingredient, DMI. Section~\ref{sec:dmim} presents the \dmi-Mechanism and the main result for information elicitation (Theorem~\ref{thm:main}). Section~\ref{sec:rule} starts to introduce the problem of information evaluation without verification and then presents the main result for information evaluation (Theorem~\ref{thm:main2}). Section~\ref{sec:con} concludes the paper and discusses possible future directions.

\section{Related work}

Since the seminal work of \citet{MRZ05}, several works focus on designing information elicitation mechanisms without verification, i.e., peer prediction. Two main settings are concerned: single-task setting and multi-task setting. In the single-task setting, participants are asked to answer only a single question. The multi-task setting asks the participants to answer several a priori similar questions. Thus, in some sense, the multi-task setting has a stronger assumption than the single-task setting. However, the multi-task setting's mechanisms are usually minimal, which is possibly more practical than the non-minimal single-task mechanisms. Moreover, the current state-of-the-art mechanism in the single-task setting is only informed-truthful while there exist dominantly truthful mechanisms in the multi-task setting. We will classify the related works into these two settings. This section will first focus on the classic setting where each task is simply a multi-choice question and each participant receives a single private answer for each task. We will talk about other models later.

\paragraph{Single-task} The original peer prediction work \cite{MRZ05} designs the first single-task peer prediction mechanism where truth-telling is a strict equilibrium. However, their mechanism requires the prior knowledge. \citet{prelec2004bayesian} designs the first single-task \emph{detail-free} peer prediction mechanism, Bayesian Truth Serum (BTS). However, BTS requires the number of participants to be infinite and is \emph{non-minimal}, in the sense that participants are additionally asked to forecast their peers' answers (e.g. what percentage of your peers will answer ``yes'' for this question?). A series of works (e.g. \cite{radanovic2014incentives,faltings2014incentives,witkowski2012robust}) focus on removing the infinite-participants assumption of BTS but are still non-minimal. \citet{DBLP:conf/innovations/KongS18} propose a detail-free mechanism which only requires a small number of participants (in fact, $\geq 6$) and pays truth-telling better than any other symmetric equilibrium. For the information evaluation, \citet{prelec2017solution} apply BTS to recover ground truth when the majority can be wrong sometimes. Unlike these works, this paper focuses on the minimal multi-task setting.

\paragraph{Multi-task} Table \ref{table:compare} compares the multi-task peer prediction mechanisms for the number of tasks needed to achieve the (approximated) truthfulness. \citet{dasgupta2013crowdsourced} propose the first multi-task peer prediction mechanism. Their mechanism is informed-truthful and requires only $\geq 2$ questions. However, their mechanism only works for binary-choice questions and assumes that all participants' honest answers are positively correlated, which put a limitation on the prior. Correlated Agreement (CA) mechanism~\cite{2016arXiv160303151S} extends \citet{dasgupta2013crowdsourced} to a non-binary setting. CA also removes the limitation on the prior. But to be informed-truthful with only a small number of questions, CA requires some prior knowledge and is not detail-free. A detail-free version of CA achieves informed truthfulness with an infinite number of questions and approximated informed truthfulness with a finite number of questions (Table~\ref{table:compare}). Although the original paper did not claim, the detail-free version of CA is also dominantly truthful with an infinite number of questions. \citet{Kong:2019:ITF:3309879.3296670} independently achieve the dominant truthfulness with an infinite number of questions and propose an information-theoretic approach. \citet{kamble2015truth} propose a mechanism where each participant can perform only one task. However, this mechanism still requires the total number of tasks to be large. Moreover, this mechanism has a weak truthfulness property: truth-telling is an equilibrium that is only better than any symmetric equilibrium where all participants perform the same strategy. In their mechanism, each participant can perform one task but the number of tasks performed by all participants should still be large. \citet{liuchen} borrow ideas from the machine learning literature on learning with noisy data and also propose a dominantly truthful mechanism. They then apply their mechanism to evaluate the information quality and recover the unknown ground truth, in the setting where the ground truth exists. \citet{liuchen}'s mechanism requires a large number of tasks, and some prior knowledge (Table~\ref{table:compare}). The evaluation rule in \citet{liuchen} has different properties from the current paper's. \citet{SchoenebeckY2019} show how to reduce the mechanism design problem of peer prediction to a learning problem which can be solved using empirical risk minimization. This enables them to reduce the task complexity significantly in their setting. Their techniques also extend to settings with continuous signals where the prior is a parametric distribution. However, unlike the current paper, they only obtain approximated truthfulness, and, in fact, they show this is a necessary limitation of their technique.

To the best of our knowledge, this paper proposes a mechanism, \dmi-Mechanism, which is both the first detail-free informed-truthful and the first detail-free dominantly truthful mechanism that requires a finite number of tasks, not to say a constant number of tasks.

    \begin{table*}[htb] \label{table:compare}
\centering
\begin{tabular}{lcc} 
    \toprule
     & ($\epsilon,\delta$)-/(0,0)-Informed-Truthful  & ($\epsilon,\delta$)-/(0,0)-Dominantly Truthful \\ 
    \toprule
   DG 2013~\cite{dasgupta2013crowdsourced} & \multirow{2}{*}{2} & \multirow{2}{*}{N/A}  \\
   (Binary-choice, limited priors) & & \\
   \midrule
    CA Mechanism \cite{2016arXiv160303151S}\footnotemark & $O(\frac{-\log \delta}{\epsilon^2})$/Infinite & (N/A)/Infinite \\ 
    \midrule
    $f$-Mutual Information \cite{Kong:2019:ITF:3309879.3296670}& \multirow{2}{*}{(N/A)/Infinite} & \multirow{2}{*}{(N/A)/Infinite}  \\ 
    Mechanism & & \\
    \midrule
     Dominant Truthful Serum~\cite{liuchen} & \multirow{2}{*}{N/A} & \multirow{2}{*}{$O(\frac{-\log \delta}{\epsilon^2})$/Infinite} \\
      (Non-detail-free) & & \\
    \midrule
    \textbf{DMI-Mechanism}  & \bm{$2C$} & \bm{$2C$} \\
    \bottomrule
\end{tabular}
\caption{A Task Sample Complexity Comparison of Multi-task Peer Prediction}
\label{table:distinguishers}

\end{table*}

\footnotetext{A non-detail-free version of CA also achieves informed truthfulness with 2 questions. Moreover, although the original paper of CA \cite{2016arXiv160303151S} does not claim the dominant truthfulness, its detail-free version is dominantly truthful.}

\paragraph{Other models} \citet{kong2018water} consider the setting where the mechanism elicits the forecasts from the crowds and each task is a forecast elicitation question (e.g. what is the probability that this start-up will succeed?). In the multi-task version of this setting, \citet{kong2018water} propose an informed-truthful mechanism. \citet{2016arXiv160607042G} propose a setting where spurious ``cheap signals'' exist (e.g. the participants collude to answer the first letter of the question, without reading the question). \citet{Kong:2018:EEW:3219166.3219172} deal with this setting by making an additional assumption, more sophisticated agents know the beliefs of less sophisticated agents.

\section{Mechanism Design for Multi-task Peer Prediction}\label{sec:prelim1}

\paragraph{Basic Notations} For integer $C>0$, we use $[C]$ to denote $\{1,2,\cdots,C\}$ and define $\mathcal{C}:=[C]$. $\Delta_{\mathcal{C}}$ is the set of all distributions over $\mathcal{C}$. $\pi:\mathcal{C}\mapsto\mathcal{C}$ is a \emph{permutation} over $\mathcal{C}$ if for every $c,c'\in\mathcal{C}$, $c\neq c'$, $\pi(c)\neq \pi(c')$. For every $C\times C$ matrix $\mathbf{A}$, $\det(\mathbf{A})$ represents $\mathbf{A}$'s determinant; $\mathrm{per}(\mathbf{A})$ represents $\mathbf{A}$'s permanent. A $C\times C$ matrix $\mathbf{A}$ is a \emph{transition matrix} if all entries of $\mathbf{A}$ is in $[0,1]$ and every row of $\mathbf{A}$ sums to 1. A $C\times C$ matrix $\mathbf{A}$ is a \emph{permutation matrix} if there exists a permutation $\pi:\mathcal{C}\mapsto\mathcal{C}$ such that $\mathbf{A}(c,\pi(c))=1,\forall c$ and all other entries are zero.  


In this section, we will introduce the multi-task peer prediction setting as well as the definition for multi-task peer prediction mechanisms. We will also introduce the formal definitions of strategy, truth-telling, permutation strategy and informed truthfulness, dominant truthfulness in the multi-task setting. 

\paragraph{Multi-task Information Elicitation} 
There are $n\geq 2$ agents. Each agent will be assigned $T$ tasks, in fact, $T$ multi-choice questions with the choice space $\mathcal{C}=[C]$ where $C$ is the number of choices (e.g. for binary-choice question, $\mathcal{C}=\{1,2\}$). For each task $t$, each agent $i$ will receive a private signal $c_i^t\in \mathcal{C}$. All agents' private signals for task $t$ are drawn from an unknown prior distribution $U_{[n]}^t\in\Delta_{\mathcal{C}^n}$. Before receiving task $t$, by denoting the unknown private signals agents will receive for task $t$ via random variables $X_1^t,X_2^t,\cdots,X_n^t$, we have 
\[\Pr[X_1^t=c_1^t, X_2^t=c_2^t,\cdots,X_n^t=c_n^t]= U_{[n]}^t(c_1^t,c_2^t,\cdots,c_n^t).\] 

Although for the same task, agents' private signals are correlated, the tasks are independent of each other, that is, $(X_i^1)_i$, $(X_i^2)_i$, ..., $(X_i^T)_i$ are independent. 

The information elicitation part's analysis focuses on the setting \emph{where agents do not need to invest efforts to obtain the private signals} (e.g. restaurant reviews). The mechanism can still be applied to the setting that considers effort, while the effort setting will have a different solution concept (for example, the dominant truthfulness will be redefined since it's impossible to incentivize effort when other people do not invest any effort).

\begin{assumption}[A priori similar tasks]\label{assume:apriori}
We assume that there exists an unknown prior distribution $U_{[n]}\in\Delta_{\mathcal{C}^n}$ such that $U_{[n]}^t = U_{[n]},\forall t$. \end{assumption}

This assumption assumes that all tasks look similar before receiving them. With this assumption, $\{(c_1^t,c_2^t,\cdots,c_n^t)\}_{t=1}^{T}$ can be seen as i.i.d. samples that are drawn from $U_{[n]}$ and we can remove the superscript $t$ and use random variables $X_1,X_2,\cdots,X_n$ to denote the unknown private signals agents will receive for every task $t$.

Unlike the single-task peer prediction setting where the common prior and the homogeneous prior assumption are commonly assumed~\cite{MRZ05,prelec2004bayesian}, in the multi-task setting, we allow agents to believe different priors (no common prior assumption) as long as all agents know that the tasks are a priori similar. For example, agent $i$ can think the prior is $U'_{[n]}$ and agent $j$ can think the prior is $U''_{[n]}$, where $U'_{[n]}$ may not equal $U''_{[n]}$. Moreover $U_{[n]}$ is allowed to be asymmetric (heterogeneous priors are allowed). For example, in the two agents setting, it's fine that $U_{[2]}(\text{yes},\text{no})$ may not equal $U_{[2]}(\text{no},\text{yes})$.

\paragraph{Multi-task Peer Prediction Mechanism}
A multi-task peer prediction mechanism will ask agents to report their private signals for $T$ tasks and collect their reports $\{(\hat{c}_1^t,\hat{c}_2^t,\cdots,\hat{c}_n^t)\}_{t=1}^{T}$ and pay each agent based on the collected reports. Formally, 

\begin{definition}[Mechanism]
We define a multi-task peer prediction mechanism $\mathcal{M}$ as a mapping $\mathcal{M}:(\mathcal{C}^n)^T \mapsto \mathbb{R}^n$ that maps from agents' reports $\{(\hat{c}_1^t,\hat{c}_2^t,\cdots,\hat{c}_n^t)\}_{t=1}^{T}$ to agents' payments $(p_1,p_2,...,p_n)$.
\end{definition}

Agents may not tell the truth and perform some strategies such that $\{(\hat{c}_1^t,\hat{c}_2^t,\cdots,\hat{c}_n^t)\}_{t=1}^{T}$ is not always equal to $\{(c_1^t,c_2^t,\cdots,c_n^t)\}_{t=1}^{T}$. Typically, the strategy of each agent should be a mapping from her received knowledge including her prior and her private signal, to a probability distribution over her report space $\mathcal{C}$. But since all agents' priors are fixed when they play the mechanism, without loss of generality, we omit the prior in the definition of strategy.

\begin{definition}[Strategy]
Given a multi-task peer prediction mechanism $\mathcal{M}$, we define each agent $i$'s strategy for each task $t$ as a $S_i^t:\mathcal{C}:\mapsto \Delta_{\mathcal{C}}$ such that given agent $i$ receives private $c_i^t$, she will randomly draw a signal $\hat{c}_i^t$ according to the distribution $S_i^t(c_i^t)$. \end{definition}

Every strategy $S_i^t$ corresponds to a $C\times C$ transition matrix $\mathbf{S}_i^t$ where $\mathbf{S}_i^t(c_i^t,\hat{c}_i^t)$ is the probability that agent $i$ reports $\hat{c}_i^t$ given that she receives private signal $c_i^t$. Agent $i$ plays a \emph{truthful} strategy if for every task $t$, $\mathbf{S}_i^t$ is an identity matrix. Agent $i$ plays a \emph{permutation} strategy if there exists a permutation matrix $\mathbf{P}$ such that for every task $t$, $\mathbf{S}_i^t=\mathbf{P}$. Truthful strategy is also a permutation strategy. 

Recall that we use random variables $X_1,X_2,\cdots,X_n$ to denote the private signals agents will receive for every task $t$. We omit the superscript $t$ due to the a priori similar tasks assumption. For each task $t$, we use $\hat{X}_1^t,\hat{X}_2^t,\cdots,\hat{X}_n^t$ to denote the signals agents will report for task $t$. The distribution of $\hat{X}_1^t,\hat{X}_2^t,\cdots,\hat{X}_n^t$ depends not only on the distribution of $X_1,X_2,\cdots,X_n$ but also on agents' strategies. We will remove the superscript $t$ with the following assumption.


\begin{assumption}[Consistent strategy] \label{assume:consistent}
    We assume every agent $i$ plays the same strategy $S_i$ for all tasks. \end{assumption}

This assumption is reasonable when agents receive tasks in independent random orders. \citet{2016arXiv160303151S} show that the above assumption can be removed when the mechanism is linear in the joint distribution over agents' reports. However, the mechanism designed in this paper is not linear. Thus, the consistent strategy assumption is assumed here. With this assumption, we can use $\hat{X}_1,\hat{X}_2,\cdots,\hat{X}_n$ to denote the signals agents will report for every task $t$. We define a \emph{strategy profile} $\mathbf{S}$ as $(S_1,S_2,\cdots,S_n)$.

A \emph{Bayesian Nash equilibrium} consists of a strategy profile $\mathbf{S}$ such that no agent wishes to change her strategy since other strategies will decrease her expected payment, given the strategies of the other agents and the information contained in her prior.

We start to present the mechanism design goals. We first define informative peer. Informally, two agents are informative peers if their private signals are ``strictly correlated''. 

\begin{definition}[Informative peer]
    Agent $i$ and agent $j$ are each other's informative peer if $\det(\mathbf{U}_{X_i,X_j})\neq 0$, where $U_{X_i,X_j}$ denotes the prior distribution over $X_i$ and $X_j$. $\mathbf{U}_{X_i,X_j}$ is the matrix format of $U_{X_i,X_j}$.  
\end{definition}

The first design goal is to incentivize the rational agents to participate in the game, with some mild conditions. We then formally define the informed truthfulness and the dominant truthfulness. 

\begin{definition}[Individual rationality]
    A mechanism is individually rational if each agent's expected payment is strictly positive if she tells the truth and at least one of her informative peers participates and tells the truth. 
\end{definition}

\begin{definition}[Informed truthfulness]
A mechanism $\mathcal{M}$ is informed-truthful if the strategy profile where everyone playing truthfully 1) is a \emph{Bayesian Nash equilibrium}; and 2) pays everyone higher than any non-truthful strategy profile and when each agent has at least one informative peer, pays everyone strictly higher than any uninformative strategy profile where $\hat{X}_1,\hat{X}_2,\cdots,\hat{X}_n$ are independent.
\end{definition}

\begin{definition}[Dominant truthfulness]
A mechanism $\mathcal{M}$ is dominantly truthful if 1) for every agent, truthful strategy maximizes her expected payment no matter what strategies other agents play; and 2) if she believes at least one of her informative peers will tell the truth, then truthful strategy pays her \emph{strictly} higher than any non-$\full$ strategy.
\end{definition}

The naive flat-payment mechanism is not informed-truthful nor dominantly truthful since it does not satisfy the second requirement of both definitions. A mechanism is dominantly truthful does not mean that it is informed-truthful. In a dominantly truthful mechanism, truth-telling is an equilibrium but may not be the best strategy profile (see \citet{Kong:2019:ITF:3309879.3296670} for a concrete counterexample). Thus, dominant truthfulness does not dominate informed truthfulness. The next section will show that \dmi-Mechanism is not only dominantly truthful but also informed-truthful. 

\section{Determinant based Mutual Information}\label{sec:dmi}

This section introduces the main technical ingredient, a novel information measure, Determinant based Mutual Information (DMI). DMI is a generalization of Shannon's mutual information. Like Shannon mutual information (MI) \cite{Shannon:2001:MTC:584091.584093}, DMI is non-negative, symmetric and also satisfies data processing inequality. In the next section, we will see it's easy to construct an unbiased estimator of the square of DMI with a constant number of samples, which is different from MI who has a log function in the formula. This is why that we can use DMI to construct a dominantly truthful mechanism that works for a constant number of tasks.  

For two random variables $X,Y$ which have the same support $\mathcal{C}$, we use $U_{X,Y}$ to denote the joint distribution over $X$ and $Y$, i.e., \[U_{X,Y}(x,y)=\Pr[X=x,Y=y].\] We use $U_{Y|X}$ to denote the transition distribution between $X$ and $Y$, i.e. \[U_{Y|X}(x,y)=\Pr[Y=y|X=x].\] We use $\mathbf{U}_{X,Y}$ and $\mathbf{U}_{Y|X}$ to denote the $C\times C$ matrix format of $U_{X,Y}$ and $U_{Y|X}$.

\begin{definition}[DMI]\label{def:dmi}
    Given two random variables $X,Y$ which have the same support $\mathcal{C}$, we define the determinant mutual information between $X$ and $Y$ as $$\dmi(X;Y)=|\det(\mathbf{U}_{X,Y})|.$$
\end{definition}

When both $X$ and $Y$ are binary (0 or 1), by simple calculations, |\dmi(X;Y)| is proportional to the classic correlation formula $|\E X Y-\E X\E Y|$. For non-binary variables, \dmi(X;Y) gives a new information measure. 

\begin{lemma}\label{lem:dmi}
$\dmi(\cdot;\cdot)$ satisfies
    \begin{description}
    \item [Symmetry] $\dmi(X;Y)=\dmi(Y;X)$;
    \item [Non-negativity and boundedness] for all $X$ and $Y$, $\dmi(X;Y)$ is in $[0,(\frac{1}{C})^C]$ and when $X$ is independent of $Y$, $\dmi(X;Y)=0$; 
      
\item [(Strict) Information-Monotonicity] for every two random variables $X,Y$ which have the same support $\mathcal{C}$, when $X'$ is less informative than $X$, i.e., $X'$ is independent of $Y$ conditioning $X$,  $\dmi(X';Y)\leq \dmi(X;Y)$. The inequality is strict when $\det(\mathbf{U}_{X,Y})\neq 0$ and $\mathbf{U}_{X'|X}$ is not a permutation matrix.  \item [Relatively Invariance] in the above setting, $\dmi(X';Y)=\dmi(X;Y)|\det(\mathbf{U}_{X'|X})|$. 
  \end{description}
\end{lemma}

\begin{proof}

DMI is symmetric since $\det(\mathbf{A}^{\top})=\det(\mathbf{A})$. DMI is non-negative since it is defined as the absolute value of the joint distribution matrix's determinant and when $X$ and $Y$ are independent, $\mathbf{U}_{X,Y}$ is rank one thus it has zero determinant. 

We start to show the relative invariance and the information-monotonicity. Note that when $X'$ is independent of $Y$ conditioning $X$, \[U_{X',Y}(x',y)=\Pr[X'=x',Y=y]=\sum_x \Pr[X'=x'|X=x]\Pr[X=x,Y=y].\] Thus, $\mathbf{U}_{X',Y}=\mathbf{U}_{X'|X}\mathbf{U}_{X,Y}$. Then

\begin{align*}
\dmi(X';Y)
=& |\det(\mathbf{U}_{X',Y})| \\ \tag{$\mathbf{U}_{X',Y}=\mathbf{U}_{X'|X}\mathbf{U}_{X,Y}$}
=& |\det(\mathbf{U}_{X'|X}\mathbf{U}_{X,Y})|\\ \tag{$\det(\mathbf{AB})=\det(\mathbf{A})\det(\mathbf{B})$}
= & \dmi(X;Y) |\det(\mathbf{U}_{X'|X})|\\ 
\leq & \dmi(X;Y)
\end{align*}
The above formula shows the relative invariance and the information-monotonicity. The last inequality holds since for every square transition matrix $\mathbf{T}$, $|\det(\mathbf{T})|\leq 1$ and the equality holds if and only if $\mathbf{T}$ is a permutation matrix (Fact~\ref{fact:dmi} \cite{seneta2006non})\footnote{Appendix~\ref{appendix:a} also presents a direct and basic proof for this fact as a reference.}. Thus, the inequality is strict when $\det(\mathbf{U}_{X,Y})\neq 0$ and $\mathbf{U}_{X'|X}$ is not a permutation matrix. 

Finally, we show that $\dmi(X;Y)$ is bounded by $(\frac{1}{C})^C$ for all $X$ and $Y$. 
\begin{align*}
	\dmi(X;Y)=&|\det(\mathbf{U}_{X,Y})| \\
	=& |\det(\mathbf{U}_{X|Y}\mathbf{U}_{Y,Y})|\\
	\leq & |\det(\mathbf{U}_{Y,Y})|\leq  (\frac{1}{C})^C
\end{align*}
\end{proof}

\section{\dmi-Mechanism}\label{sec:dmim}

The original idea of peer prediction~\cite{MRZ05} is based on a clever insight: every agent's information is related to her peers' information and therefore can be checked using her peers' information. Inspired by this, \citet{Kong:2019:ITF:3309879.3296670} propose a natural yet powerful information-theoretic mechanism design idea---paying every agent the ``mutual information'' between her reported information and her peer's reported information where the ``mutual information'' should be \emph{information-monotone}---any ``data processing'' on the two random variables will decrease the ``mutual information'' between them. As we assume the agents want to maximize their payments in expectation, it's sufficient to design mechanisms such that the payment in the mechanism is an unbiased estimator of the information-monotone measure.

However, the unbiased estimator of information-monotone measure used in \citet{Kong:2019:ITF:3309879.3296670}, $f$-mutual information and Bregman-mutual information, cannot be constructed with a constant number of samples. The main technical reason is that those measures do not have a polynomial format. Unlike the previous measures, DMI's square has a polynomial format such that its unbiased estimator can be constructed with a constant number of samples, in fact, $2C$ samples. With this nice property of DMI, we propose a novel mechanism, \dmi-Mechanism, based on DMI such that \dmi-Mechanism is dominantly truthful with only a constant number of tasks. 

\begin{figure*}[h!]
    \centering
    \includegraphics[width=5.5in]{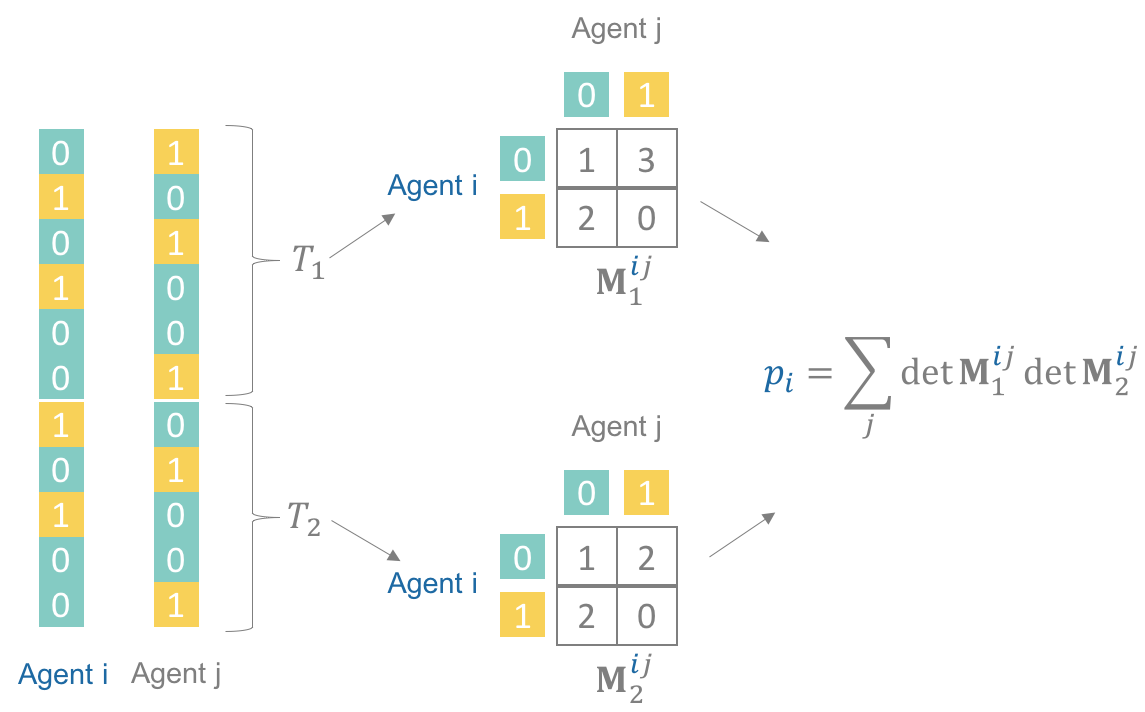}
    \caption{The illustration of the \dmi-Mechanism: we use binary-choice questions to illustrate the mechanism, while this mechanism works for general multi-choice questions. The leftmost part is agent $i$ and agent $j$'s answers for $T=11$ binary-choice questions. The $T$ questions are divided into two parts, $T_1$ and $T_2$. Each $T_{\ell}$ corresponds to a $2\times 2$ matrix $\mathbf{M}_{\ell}^{ij}$ such that for each $(c,c')\in \{0,1\}^2$, $\mathbf{M}_{\ell}^{ij}(c,c')$ counts the number of tasks in $T_{\ell}$ that agent $i$ answers $c$ and agent $j$ answers $c'$. With these two matrices, agent $j$'s contribution to agent $i$'s payment is defined as $\det \mathbf{M}_1^{ij}\det \mathbf{M}_2^{ij}=-6\times (-4)=24$. In this running example, agent $i$ and agent $j$'s answers are actually negatively correlated such that both $\det \mathbf{M}_1^{ij}$ and $\det \mathbf{M}_2^{ij}$ are negative, while $\det \mathbf{M}_1^{ij}\det \mathbf{M}_2^{ij}$ is still positive. Thus, when agent $i$ and agent $j$'s true answers are negatively correlated, for agent $i$, truth-telling is still not worse than permuting her answers to match agent $j$ in \dmi-Mechanism. } 
    \label{fig:ex}
\end{figure*}

\paragraph{\dmi-Mechanism} $n$ agents are assigned $T\geq 2C$ a priori similar tasks in independent random orders. Agents finish the tasks without any communication. 
\begin{description}
\item[Report] For each task $t$, each agent $i$ privately receives $c_i^t$ and reports $\hat{c}_i^t$. 
\item[Payment] $T$ tasks are arbitrarily divided into two disjoint parts $T_1,T_2$ such that both $|T_1|,|T_2|\geq C$. For every two agents $i\neq j\in [n]$, we define two $C\times C$ \emph{answer matrices} $\mathbf{M}_1^{ij}, \mathbf{M}_2^{ij}$ where for $\ell=1,2$ \[\mathbf{M}_{\ell}^{ij}(c,c'):=\sum_{t\in T_{\ell}} \mathbbm{1}\left( (\hat{c}_i^t, \hat{c}_j^t )=(c,c')\right)\]

Agent $i$'s payment is 
\[ p_i:= \sum_{j\neq i\in [n]}\det(\mathbf{M}_1^{ij})\det(\mathbf{M}_2^{ij}) \]

\end{description}

\paragraph{Normalization of DMI-mechanism} In the above definition of DMI-mechanism, the payment depends on the number of participants and the number of tasks, which may induce unbounded payments. This issue can be addressed by normalizing the payments:

\[ q_i=\mathrm{norm}(p_i):=\frac{p_i}{(n-1)*(C!)^2 {|T_1|\choose C}{|T_2|\choose C}}.\] The following theorem will show that the normalized payment is in $[0,(\frac{1}{C})^C]$ in expectation. Since any linear transformation $a p_i + b, a>0,b\geq 0$ on the payments will not change the truthful and rational properties, the normalized DMI-mechanism will also be detail-free, minimal, individually rational, informed-truthful and dominantly truthful.

\begin{theorem}[Main Theorem]\label{thm:main}
When $n\geq 2$ and $T\geq 2 C$, \dmi-Mechanism is detail-free, minimal, individually rational, informed truthful, and dominantly truthful. In the normalized DMI-mechanism, every participant's expected payment is in $[0,(\frac{1}{C})^C]$. 
\end{theorem}

Note that $\frac{1}{T_{\ell}}\mathbf{M}_{\ell}^{ij}(c,c')$ is an unbiased estimator of $\Pr[\hat{X}_i=c,\hat{X}_j=c']$ (see preliminaries for the definition of $\hat{X}_i,\hat{X}_j$). Thus, a naive thought would be to show that $(\frac{1}{T_{\ell}})^{C}\det(\mathbf{M}_{\ell}^{ij})$ is an unbiased estimator of $\det(\mathbf{U}_{\hat{X}_i ,\hat{X}_j})$, where $U_{\hat{X}_i,\hat{X}_j}$ is the joint distribution over $\hat{X}_i,\hat{X}_j$. However, since the entries of $\mathbf{M}_{\ell}^{ij}$ are not independent, the naive thought is not valid. Luckily, by carefully analyzing the correlation among the entries of $\mathbf{M}_{\ell}^{ij}$ and writing down the explicit formula of determinant, we will see $\det(\mathbf{M}_{\ell}^{ij})$ can be written as a sum of indicator variables and each indicator variable corresponds $C$ ordered distinct tasks in $T_{\ell}$. With this observation, the native thought is true by replacing the constant $(\frac{1}{T_{\ell}})^{C}$ by one over the number of ordered $C$ distinct tasks (in fact, $\frac{1}{{|T_{\ell}|\choose C} C!}$). Moreover, this observation also implies that every $|T_{\ell}|$ must be greater than $C$ to have at least $C$ distinct tasks such that the expected payment can be non-zero. 

\begin{proof}

With the assumption that that \dmi-Mechanism assign a priori similar tasks (Assumption~\ref{assume:apriori}) and agents play consistent strategies (Assumption~\ref{assume:consistent}), we can use $\hat{X}_1,\hat{X}_2,\cdots,\hat{X}_n$ to denote the signals agents will report for every task and $\{(\hat{c}_1^t,\hat{c}_2^t,\cdots,\hat{c}_n^t)\}_{t=1}^{T}$ are $T$ i.i.d. samples of $\hat{X}_1,\hat{X}_2,\cdots,\hat{X}_n$. We first claim that:

\begin{claim}\label{claim:dmi}
    The expected $\det(\mathbf{M}_{\ell}^{ij})$ is proportional to $\det(\mathbf{U}_{\hat{X}_i ,\hat{X}_j})$, that is, there exists a constant $a_{\ell}$ such that $\det(\mathbf{M}_{\ell}^{ij})=a_{\ell} \det(\mathbf{U}_{\hat{X}_i ,\hat{X}_j})$. In fact, $a_{\ell}={|T_{\ell}|\choose C} C!$.
\end{claim}

Recall that the reports are i.i.d. samples. Then $\det(\mathbf{M}_1^{ij})$ and $\det(\mathbf{M}_2^{ij})$ are independent, since they are related to two distinct sets of tasks. Thus, with the above claim, the expected payment of agent $i$ will be

\[ \sum_{j\neq i}a_1 a_2 \dmi^2(\hat{X}_i;\hat{X}_j).\] 

This implies that the normalized payment is in $[0,(\frac{1}{C})^C]$ in expectation. Note that agents finish the tasks without communication, which implies that agents play strategies independently. Thus, $\hat{X}_i$ is less informative than $X_i$, i.e., conditioning on $X_i$, $\hat{X}_i$ is independent of other agents' reports. Therefore, due to the information-monotonicity of DMI, 

\begin{align*}
    \sum_{j\neq i}a_1 a_2 \dmi^2(\hat{X}_i;\hat{X}_j) \leq & \sum_{j\neq i}a_1 a_2 \dmi^2(X_i;\hat{X}_j)\\
    \leq & \dmi^2(X_i;X_j)
\end{align*}

The first inequality shows that truth-telling is a dominant strategy, which also implies that everyone playing truthfully is an equilibrium. The second inequality shows that everyone playing truthfully pays everyone better than any non-truthful strategy profile. It's left to prove the strictness conditions. 

When one of agent $i$'s informative peers, call him agent $j$, reports truthfully and agent $i$ plays a non-permutation strategy, we have 
\begin{align*}
\dmi^2(\hat{X}_i;\hat{X}_j)=&\dmi^2(\hat{X}_i;X_j)\\
<&\dmi^2(X_i;X_j)\\
=&\dmi^2(X_i;\hat{X}_j)
\end{align*}

The strict inequality follows from the strict information-monotonicity of DMI. This shows that \dmi-Mechanism satisfies the strictness condition of dominant truthfulness. For the strictness condition of informed truthfulness, due to the non-negativity of \dmi, each agent's expected payment is zero in all uninformative strategy profiles and for each agent $i$, if she plays truthfully and one of her informative peers plays truthfully as well, agent $i$'s expected payment will be strictly positive. Therefore, \dmi-Mechanism is individually rational, dominantly truthful and informed truthful. It remains to show the above claim. 

We define $I_t^{ij}(c,c'):=\mathbbm{1}\left( (\hat{c}_i^t, \hat{c}_j^t )=(c,c')\right)$ as an indicator random variable that indicates whether agent $i$ and agent $j$'s answers for task $t$ is $(c,c')$. 

\begin{align*} 
&\det(\mathbf{M}_{\ell}^{ij})\\
=&  \sum_{\pi}\sign(\pi) \prod_{c\in\mathcal{C}} \mathbf{M}_{\ell}^{ij}(c,\pi(c))\\
=&  \sum_{\pi}\sign(\pi) \prod_{c\in\mathcal{C}} \sum_{t\in T_{\ell}} I_t^{ij}(c,\pi(c))\\
=&  \sum_{\pi}\sign(\pi) \prod_{c\in\mathcal{C}} \sum_{t_{c}\in T_{\ell}} I_{t_{c}}^{ij}(c,\pi(c))\\
=& \sum_{\pi}\sign(\pi) \sum_{t_1\in T_{\ell}}\cdots\sum_{t_{c}\in T_{\ell}}\cdots\sum_{t_{C}\in T_{\ell}} \prod_{c}  I_{t_{c}}^{ij}(c,\pi(c))\\
=& \sum_{\pi}\sign(\pi) \sum_{(t_1,\cdots,t_{c},\cdots,t_{C})\in T_{\ell}^{C}} \prod_{c}  I_{t_{c}}^{ij}(c,\pi(c))\\ 
\end{align*}

It's sufficient to only consider distinct tasks. We use $\mathbf{t}_{\ell}=(\mathbf{t}_{\ell}(1),\mathbf{t}_{\ell}(2),\cdots,\mathbf{t}_{\ell}(C))$ to denote ordered $C$ distinct tasks in $T_{\ell}$ and $I_{\mathbf{t}_{\ell}}^{ij}(\pi):=\Pi_c I_{\mathbf{t}_{\ell}(c)}^{ij}(c,\pi(c))$. Note that for indistinct $(t_1,\cdots,t_{c},\cdots,t_{C})$, e.g. $t_1=t_2$, one of $I_{t_1}^{ij}(1,\pi(1))$ and $I_{t_2}^{ij}(2,\pi(2))$ must be zero. Thus, 
\[ \det(\mathbf{M}_{\ell}^{ij})= \sum_{\pi}\sign(\pi)\sum_{\mathbf{t}_{\ell}} I_{\mathbf{t}_{\ell}}^{ij}(\pi)\]

We use $b_{\pi}$ to denote the expectation of $I_{\mathbf{t}_{\ell}}^{ij}(\pi)$, then since distinct tasks are independent,
\begin{align*}
    b_{\pi}:=& \E_{\hat{X}_i ,\hat{X}_j} I_{\mathbf{t}_{\ell}}^{ij}(\pi)\\
           = & \E_{\hat{X}_i ,\hat{X}_j} \Pi_c I_{\mathbf{t}_{\ell}(c)}^{ij}(c,\pi(c))\\ 
           = & \prod_{c} \Pr[\hat{X}_i =c, \hat{X}_j=\pi(c)].\\
           = & \prod_{c} \mathbf{U}_{\hat{X}_i,\hat{X}_j}(c,\pi(c))
\end{align*}

Thus, 

\begin{align*}
&\E_{\hat{X}_i ,\hat{X}_j}(\mathbf{M}_{\ell}^{ij})\\=& \sum_{\pi}\sign(\pi)\sum_{\mathbf{t}_{\ell}} b_{\pi}\\ 
=& \sum_{\mathbf{t}_{\ell}}  \sum_{\pi}\sign(\pi)b_{\pi}\\
= & \sum_{\mathbf{t}_{\ell}}  \sum_{\pi}\sign(\pi) \prod_{c} \mathbf{U}_{\hat{X}_i,\hat{X}_j}(c,\pi(c)) = a_{\ell} \det(\mathbf{U}_{\hat{X}_i ,\hat{X}_j})
\end{align*}

where $a_{\ell}={|T_{\ell}|\choose C} C!$. We finish the proof of Claim~\ref{claim:dmi}. 

\end{proof}

\section{DMI-Mechanism as a Scoring Rule}\label{sec:rule}

\dmi-Mechanism can be naturally transferred into a scoring rule. This scoring rule scores every participant via her payment in the \dmi-Mechanism. This section will show that this \dmi-Mechanism based scoring rule scores information without ground truth, but in expectation gives scores like it has the full knowledge of ground truth. Before the introduction of the formal setting, let's first see the following running example, which will help understand the formal setting. 

\begin{example} [peer grading]
    $n$ students are asked to grade $T$ a priori similar essays and the score space is $\mathcal{C}=\{\text{good, okay, fail}\}$. Each essay $t$ has an unknown ground truth score $g^t\in \mathcal{C}$, which is expensive to access. We assume that for every essay, the students' reported scores are independent conditioning on the essay's ground truth score. We aim to evaluate each student's true grading quality without access to the ground truth scores and each student's true grading quality should reflect the correlation between her reported score and the ground truth score. 
\end{example}

\paragraph{Multi-task Information Evaluation}
We modify the multi-task information elicitation setting by adding an unknown ground truth. For each task $t$, there exists an unknown ground truth signal $g^t\in \mathcal{C}$. Before receiving task $t$, we denote the unknown ground truth for this task by random variable $G^t$. We use $U_{[n],G}^t$ to denote the joint distribution over all agents' private signals and the unknown ground truth for task $t$. The a priori similar assumption and the consistent strategy assumption also extend to this setting and thus we can omit the superscript $t$ and assume i.i.d. samples. 

\begin{assumption}[Conditional independence]
    We assume that agents' private signals $X_1,X_2,\cdots X_n$ are independent conditioning on $G$. Since agents' strategies are independent, this also implies that agents' reports $\hat{X}_1,\hat{X}_2,\cdots \hat{X}_n$ are independent conditioning on $G$. 
\end{assumption}
With the full knowledge of ground truth $G$, we define the quality score of each agent $i$'s reported signals by its DMI based correlation with $G$,

\[\mathrm{Quality}(\hat{X}_i):=\dmi^2[\hat{X}_i;G].\]

$X$ is \emph{informative} if it has a non-zero quality score. Informative information can still be very noisy (e.g. $X=G$ w.p. 0.01, $X=1$ w.p. 0.99). We say an agent is informative if she reports informative information. The following observation follows directly from the information-monotonicity of $\dmi$ and justifies the reasonableness of this quality definition.  

\begin{observation}
    $G$ has the highest quality score and when $X'$ is less informative than $X$, i.e., $X'$ is independent of $Y$ conditioning $X$, $X'$ has a lower quality score than $X$. 
\end{observation}

However, every information-monotone measure satisfies the above observation. The reason to use DMI here is that \emph{we can measure the DMI based quality score without the knowledge of ground truth} (see the following theorem). This is due to the relative invariance of DMI. The traditional MI does not have this special property. 

\paragraph{DMI-Mechanism based scoring rule} $n\geq 3$ agents are assigned $T\geq 2C$ a priori similar tasks in independent random orders. Agents finish the tasks without any communication. Each agent's score is her payment in DMI-Mechanism. 

\begin{theorem}[Information Evaluation]\label{thm:main2}
    When $T\geq 2C$ and the number of informative agents is $\geq 3$, the difference between every two agent's expected score in the DMI-Mechanism based scoring rule is proportional to the difference between their quality score, i.e., \[\E (p_i-p_j) \propto \mathrm{Quality}(\hat{X}_i)-\mathrm{Quality}(\hat{X}_j)\] 
    
     Moreover, by setting $|T_1|=\lc{\frac{T}{2}}\rc$ and $|T_2|=T-|T_1|$, for every $i$, and every $0<\alpha<\kappa, 0<\delta<1$, when $T\geq g(\alpha,\delta,C,\kappa):=\frac{4 C}{1-(1-\frac{\alpha^2 \delta}{3(C!+1)(n-1)\kappa })^{\frac{1}{C}}}$ where $\kappa=\max_{i\neq j} \per(\mathbf{U}_{\hat{X}_i,\hat{X}_j})$, \[ \Pr[|q_i-\E[q_i]|>\alpha]\leq \delta \] where $q_i=\mathrm{norm}(p_i)$. For every two agents $i,j$, by setting $\alpha=\frac{|\E[q_i-q_j]|}{2}$, when $T\geq g(\alpha,\frac{\delta}{2},C,\kappa)$, the probability that the DMI-Mechanism based scoring rule assigns a higher score to the lower-quality information between the two is bounded by $\delta$, i.e.,   
    \[ \Pr[(p_i-p_j)(\mathrm{Quality}(\hat{X}_i)-\mathrm{Quality}(\hat{X}_j))<0]\leq \delta \]
\end{theorem}

In the peer grading example, the above theorem shows that the student who has better grading quality will also be paid more in expectation in \dmi-mechanism. This property also shows the \emph{robustness} of \dmi-mechanism. Note that when the number of informative agents is $\geq 3$, every agent's rank is only related to the true quality of her information. Thus, to compare every two agents' information quality, the quality of other agents' information is not important, as long as the number of tasks is sufficiently large. Thus, we can know which student has a better grading capacity by just looking at her payment in \dmi-mechanism, when there are sufficient amount of tasks. 

This section will focus on the proof of the first part of the above theorem. This part's proof is natural and all results almost follow directly from the relative invariance of \dmi. The next section presents the full concentration bound analysis. Note that Claim~\ref{claim:dmi} already proved that the payment of each agent can be written as a sum of several indicator variables. Then the analysis will follow from a classic second moment method. The next paragraph gives some explanations for the task complexity. 

In addition to the error probability $\delta$, the task complexity depends on the two students' actual grading quality difference $|\E[q_i-q_j]|$, the prior information $\kappa$, the score space size $C$ and the number of students $n$. $\kappa$ is the maximal permanent over all joint distribution matrices. When $\kappa$ is small, the payment is more concentrated in an absolute way. When $\kappa$ is zero, all agents' payments must always be zero and there is no need for the concentration bound. Since there exist informative agents, $\kappa$ must be $>0$. The task complexity depends on $n$ since the proof uses a union bound for all constitutions of an agent's payment. When all agents are informative, we can always play the game within 3 agents and make $n=3$ to decrease the task complexity.   

\begin{proof} [Proof of Theorem \ref{thm:main2} Part 1]
    For every two agent $i,j$, the difference between their expected payments are 
    \begin{align*}
       & \E(p_i - p_{j}) \\
       = & \sum_{k\neq i} a_1 a_2 \dmi^2(\hat{X}_i;\hat{X}_k)-\sum_{k\neq j} a_1 a_2 \dmi^2(\hat{X}_{j};\hat{X}_k)\\ 
                         = & \sum_{k\neq i,j} a_1 a_2 \dmi^2(\hat{X}_i;\hat{X}_k)-\sum_{k\neq i, j} a_1 a_2 \dmi^2(\hat{X}_{j};\hat{X}_k)\\ \tag{DMI is symmetric} 
        =&\sum_{k\neq i,j} a_1 a_2 \left(\dmi^2(\hat{X}_i;G)-  \dmi^2(\hat{X}_{j};G)\right)\det(\mathbf{U}_{\hat{X}_k|G})^2\\ \tag{DMI is relatively invariant} 
        =& (\mathrm{Quality}(\hat{X}_i)-\mathrm{Quality}(\hat{X}_{j})) (\sum_{k\neq i,j} a_1 a_2 \det(\mathbf{U}_{\hat{X}_k|G})^2)
    \end{align*}
    
    When at least three informative agents exist, there must exist $k\neq i,j$ such that $\mathrm{Quality}(\hat{X}_k)\neq 0$, which also implies $\det(\mathbf{U}_{\hat{X}_k|G})\neq 0$ since $\mathrm{Quality}(\hat{X}_k)=(\det(\mathbf{U}_{G,G})^2\det(\mathbf{U}_{\hat{X}_k|G}))^2$ is non-zero. Thus, $(\sum_{k\neq i,j} a_1 a_2 \det(\mathbf{U}_{\hat{X}_k|G})^2)\neq 0$, which implies that \[\E (p_i-p_j) \propto \mathrm{Quality}(\hat{X}_i)-\mathrm{Quality}(\hat{X}_j).\] 
    
\end{proof}

\subsection{Concentration Bound Analysis} \label{sec:cen}

\begin{proof} [Proof of Theorem \ref{thm:main2} Part 2]

We first use union bound to reduce the proof to the analysis of the concentration bound of $q^{ij}=\mathrm{norm}(\det (\mathbf{M}_1^{ij}) \det (\mathbf{M}_2^{ij})):=\frac{\det (\mathbf{M}_1^{ij}) \det (\mathbf{M}_2^{ij})}{(C!)^2 {|T_1|\choose C}{|T_2|\choose C}} $ and then analyze $q^{ij}$ with the second moment method. 

\paragraph{Union bound reduction}
    
First note that for every $i$, since $q_i=\frac{\sum_{j\neq i}q^{ij}}{n-1}$
\begin{align*}
    \Pr[|q_i-\E[q_i]|>\alpha]\leq \Pr[\exists j, |q^{ij}-\E[q^{ij}]|>\alpha]<(n-1)\max_j \Pr[|q^{ij}-\E[q^{ij}]|>\alpha].
\end{align*} Similarly, for every two agents $i, j$,  

\begin{align*}
    & \Pr[(p_i-p_j)(\mathrm{Quality}(\hat{X}_i)-\mathrm{Quality}(\hat{X}_j))<0]\\
    \leq &\Pr[|q_i-\E[q_i]|>\frac{|\E[q_i-q_j]|}{2} \text{ or } |q_j-\E[q_j]|>\frac{|\E[q_i-q_j]|}{2}]\\
    \leq & 2 \max\{ \Pr[|q_i-\E[q_i]|>\frac{|\E[q_i-q_j]|}{2}, \Pr[|q_j-\E[q_j]|>\frac{|\E[q_i-q_j]|}{2} \}
\end{align*}

Once we prove that when $T\geq h(\alpha, \delta, C, \mathbf{U}_{\hat{X}_i,\hat{X}_j})$, $\Pr[|q^{ij}-\E[q^{ij}]|>\alpha]<\delta$, we will know that when $T\geq \max_{i,j} h(\alpha, \frac{\delta}{n-1}, C, \mathbf{U}_{\hat{X}_i,\hat{X}_j})$, $\Pr[|q_i-\E[q_i]|>\alpha]<\delta$ and when $T\geq \max_{i,j} h(\frac{|\E[q_i-q_j]|}{2}, \frac{\delta}{2(n-1)}, C, \mathbf{U}_{\hat{X}_i,\hat{X}_j})$, $Pr[(p_i-p_j)(\mathrm{Quality}(\hat{X}_i)-\mathrm{Quality}(\hat{X}_j))<0]<\delta$. Thus, it's sufficient to only analyze the concentration bound of $\max_j \Pr[|q^{ij}-\E[q^{ij}]|>\alpha]$. 

\paragraph{Writing $\det (\mathbf{M}_{\ell}^{ij})$ as a sum of indicator variables}

We fix two agents $i,j$ here. \emph{For simplicity, we omit the superscript $ij$ and will add them in the end}. That is, we use $\mathbf{U}$ to denote the joint distribution over their reports and $\mathbf{M}_1,\mathbf{M}_2$ to denote their answer matrices in DMI-mechanism. $\mathbf{t}_{\ell}=(\mathbf{t}_{\ell}(1),\mathbf{t}_{\ell}(2),\cdots,\mathbf{t}_{\ell}(C))$ denotes ordered $C$ distinct tasks in $T_{\ell}$ and $I_{\mathbf{t}_{\ell}}(\pi):=\Pi_c I_{\mathbf{t}_{\ell}(c)}(c,\pi(c))$. $b_{\pi}$ denotes the expectation of $I_{\mathbf{t}_{\ell}}(\pi)$ which is $\Pi_c \mathbf{U}(c,\pi(c))$, $f(T,C)$ denotes ${T\choose C} C!$ and $\per(\mathbf{U})=\sum_{\pi} b_{\pi}$. 

The proof in Claim~\ref{claim:dmi} shows that 
\[\det \mathbf{M}_{\ell} = \sum_{\pi} \sign(\pi)\sum_{\mathbf{t}_{\ell}\in T_{\ell}^C} I_{\mathbf{t}_{\ell}}(\pi)
\] In the next part, we will analyze $\det (\mathbf{M}_{\ell})$ via the second moment method in a straightforward manner. 

\paragraph{Analyzing $\det (\mathbf{M}_{\ell})$ via the second moment method}

Without loss of generality, we only analyze $\det (\mathbf{M}_1)$ here. 

\begin{align*}
    \var[\det \mathbf{M}_1]=&\var[\sum_{\pi}\sign(\pi)\sum_{\mathbf{t}_1} I_{\mathbf{t}_1}(\pi)]\\
    = & \sum_{\pi} \sum_{\mathbf{t}_1} \var[ I_{\mathbf{t}_1}(\pi)] + \sum_{\pi,\pi'} \sum_{\mathbf{t}_1\cap\mathbf{t}'_1\neq \emptyset} \cov(I_{\mathbf{t}_1}(\pi),I_{\mathbf{t}'_1}(\pi'))\\
    \leq & \sum_{\pi} \sum_{\mathbf{t}_1} b_{\pi} + \sum_{\pi,\pi'}\sum_{\mathbf{t}_1\cap \mathbf{t}'_1\neq \emptyset } b_{\pi} \\
    = & \sum_{\mathbf{t}_1} (\sum_{\pi} b_{\pi}) + \sum_{\pi', \mathbf{t}_1\cap \mathbf{t}'_1\neq \emptyset  } (\sum_{\pi} b_{\pi})\\
    = & f(|T_1|,C)\per(\mathbf{U}) + C! f(|T_1|,C) (f(|T_1|,C)-f(|T_1|-C,C))\per(\mathbf{U})
\end{align*}

Which induces that 
\[\var[\frac{\det \mathbf{M}_1}{f(|T_1|,C)}]\leq g(|T_1|,C) \per(\mathbf{U})\]

where 

\begin{align*} 
    g(|T_1|,C)=&\frac{1}{f(|T_1|,C)}+C!(1-\frac{f(|T_1|-C,C)}{f(|T_1|,C)})\\
    < & (C!+1) (1-\frac{f(|T_1|-C,C)}{f(|T_1|,C)})\\
    < &  (C!+1)(1-(1-\frac{2 C}{|T_1|})^C) = o(|T_1|)
\end{align*}

For \emph{independent} $X$ and $Y$, 

\begin{align*}
    \var[X Y] & = \E[X^2 Y^2] - (\E[X Y])^2\\
              & = ((\E[X])^2+\var[X])((\E[Y])^2+\var[Y])- (\E[X Y])^2\\
    &= \var[X]\var[Y]+ \var[X](\E[Y])^2 + \var[Y](\E[X])^2
\end{align*}

then

\begin{align*}
    \var[\frac{\det \mathbf{M}_1 \mathbf{M}_2}{f(|T_1|,C) f(|T_2|,C)}]&= \var[\det \mathbf{M}_1]\var[\det \mathbf{M}_1]+ \var[\det \mathbf{M}_1](\det(\mathbf{U}))^2 + \var[\det \mathbf{M}_2](\det(\mathbf{U}))^2\\
     &\leq g(|T_1|,C)g(|T_2|,C) (\per(\mathbf{U}))^2 + (g(|T_1|,C)+g(|T_2|,C))\per(\mathbf{U}) (\det(\mathbf{U}))^2
\end{align*}

Based on Chebyshev's inequality, 

\begin{align*}
    & \Pr[|\frac{\det \mathbf{M}_1 \mathbf{M}_2}{f(|T_1|,C) f(|T_2|,C)}-(\det(\mathbf{U}))^2|\geq \alpha]\\
    \leq &\frac{g(|T_1|,C)g(|T_2|,C) (\per(\mathbf{U}))^2 + (g(|T_1|,C)+g(|T_2|,C))\per(\mathbf{U}) (\det(\mathbf{U}))^2}{\alpha^2}
\end{align*}

Then when $g(|T_{\ell}|,C)<\frac{ \alpha^2\delta}{3\per(\mathbf{U})} $, $\Pr[|\frac{\det \mathbf{M}_1 \mathbf{M}_2}{f(|T_1|,C) f(|T_2|,C)}-(\det(\mathbf{U}))^2|\geq \alpha]< \delta $

Note that $g(|T_{\ell}|,C)<(C!+1)(1-(1-\frac{2 C}{|T_1|})^C)$, then we can pick $|T_l|>\frac{2C}{1-(1-\frac{ \alpha^2\delta}{3\per(\mathbf{U}) (C!+1)} )^{1/C}}$ to make $g(|T_{\ell}|,C)<\frac{ \alpha^2\delta}{3\per(\mathbf{U})} $. Thus, we can pick $T>h(\alpha, \delta, C, \mathbf{U}_{\hat{X}_i,\hat{X}_j}):= \frac{4C}{1-(1-\frac{ \alpha^2\delta}{3\per(\mathbf{U}_{\hat{X}_i,\hat{X}_j}) (C!+1)} )^{1/C}} $ (we put back the superscript $ij$) to make $\Pr[|q^{ij}-\E[q^{ij}]|>\alpha]<\delta$. To maximize over all $i,j$, we just pick $\kappa=\max_{i\neq j} \per(\mathbf{U}_{\hat{X}_i,\hat{X}_j})$. Based on the analysis in the union bound reduction part, we finish the proof. 

\end{proof}

\section{Conclusion and Discussion}\label{sec:con}
This paper proposes the first dominantly truthful and the first detail-free informed-truthful multi-task peer prediction mechanism, \dmi-Mechanism, that only requires a finite number of tasks. In fact, \dmi-Mechanism only requires a small constant number of tasks. Additionally, the mechanism is detail-free, minimal, works for non-common heterogenous priors and only requires $\geq 2$ participants. When there are $\geq 3$ informative participants, \dmi-Mechanism can also be used as a scoring rule that scores information without knowledge of ground truth. The construction of the mechanism is based on a new information measure, $\dmi$.   

For clean analysis, the information elicitation part's analysis is restricted to the setting where agents do not need to expend effort to receive their private information. In the setting where agents incur a cost for putting forth the effort required to obtain their private signal, \dmi-Mechanism still applies straightforwardly and all truthfulness properties are preserved when the analysis is restricted to the signal-report strategies. However, to incentivize efforts, each agent needs to believe at least a certain number of other agents will invest effort.

The current implementation of the mechanism asks the agents to perform the same $T$ tasks, while it's fine to ask the agents to perform different tasks as long as each agent has $\geq 2C$ overlapping tasks with other agents. With this extension, each task does not need to be performed by all agents and thus the efficiency of the mechanism will increase. 

An interesting direction is to perform real-world experiments and use \dmi-Mechanism to recover unknown ground truth. A full classification of dominantly truthful mechanisms with a small number of tasks is another interesting theoretic direction. This problem can be reduced to the search for information measures that have the polynomial format. Moreover, \dmi's weakness is that it cannot measure the correlation between random variables that have supports of different sizes. The investigation of the existence of the mechanism that can overcome this weakness and preserve all other properties is also a possible direction.  

\section*{Acknowledgements} The author would like to thank all anonymous reviewers for their careful reviews and helpful suggestions.

\bibliographystyle{plainnat}
\bibliography{reference}

\begin{thebibliography}{18}
\providecommand{\natexlab}[1]{#1}
\providecommand{\url}[1]{\texttt{#1}}
\expandafter\ifx\csname urlstyle\endcsname\relax
  \providecommand{\doi}[1]{doi: #1}\else
  \providecommand{\doi}{doi: \begingroup \urlstyle{rm}\Url}\fi

\bibitem[Dasgupta and Ghosh(2013)]{dasgupta2013crowdsourced}
Anirban Dasgupta and Arpita Ghosh.
\newblock Crowdsourced judgement elicitation with endogenous proficiency.
\newblock In \emph{Proceedings of the 22nd international conference on World
  Wide Web}, pages 319--330. International World Wide Web Conferences Steering
  Committee, 2013.

\bibitem[Faltings et~al.(2014)Faltings, Jurca, Pu, and
  Tran]{faltings2014incentives}
Boi Faltings, Radu Jurca, Pearl Pu, and Bao~Duy Tran.
\newblock Incentives to counter bias in human computation.
\newblock In \emph{Second AAAI Conference on Human Computation and
  Crowdsourcing}, 2014.

\bibitem[{Gao} et~al.(2016){Gao}, {Wright}, and
  {Leyton-Brown}]{2016arXiv160607042G}
A.~{Gao}, J.~R. {Wright}, and K.~{Leyton-Brown}.
\newblock {Incentivizing Evaluation via Limited Access to Ground Truth:
  Peer-Prediction Makes Things Worse}.
\newblock \emph{ArXiv e-prints}, June 2016.

\bibitem[Kamble et~al.(2015)Kamble, Shah, Marn, Parekh, and
  Ramachandran]{kamble2015truth}
Vijay Kamble, Nihar Shah, David Marn, Abhay Parekh, and Kannan Ramachandran.
\newblock Truth serums for massively crowdsourced evaluation tasks.
\newblock \emph{arXiv preprint arXiv:1507.07045}, 2015.

\bibitem[Kong and
  Schoenebeck(2018{\natexlab{a}})]{DBLP:conf/innovations/KongS18}
Yuqing Kong and Grant Schoenebeck.
\newblock Equilibrium selection in information elicitation without verification
  via information monotonicity.
\newblock In Anna~R. Karlin, editor, \emph{9th Innovations in Theoretical
  Computer Science Conference, {ITCS} 2018, January 11-14, 2018, Cambridge, MA,
  {USA}}, volume~94 of \emph{LIPIcs}, pages 13:1--13:20. Schloss Dagstuhl -
  Leibniz-Zentrum fuer Informatik, 2018{\natexlab{a}}.
\newblock ISBN 978-3-95977-060-6.
\newblock \doi{10.4230/LIPIcs.ITCS.2018.13}.
\newblock URL \url{https://doi.org/10.4230/LIPIcs.ITCS.2018.13}.

\bibitem[Kong and
  Schoenebeck(2018{\natexlab{b}})]{Kong:2018:EEW:3219166.3219172}
Yuqing Kong and Grant Schoenebeck.
\newblock Eliciting expertise without verification.
\newblock In \emph{Proceedings of the 2018 ACM Conference on Economics and
  Computation}, EC '18, pages 195--212, New York, NY, USA, 2018{\natexlab{b}}.
  ACM.
\newblock ISBN 978-1-4503-5829-3.
\newblock \doi{10.1145/3219166.3219172}.
\newblock URL \url{http://doi.acm.org/10.1145/3219166.3219172}.

\bibitem[Kong and Schoenebeck(2018{\natexlab{c}})]{kong2018water}
Yuqing Kong and Grant Schoenebeck.
\newblock Water from two rocks: Maximizing the mutual information.
\newblock In \emph{Proceedings of the 2018 ACM Conference on Economics and
  Computation}, pages 177--194. ACM, 2018{\natexlab{c}}.

\bibitem[Kong and Schoenebeck(2019)]{Kong:2019:ITF:3309879.3296670}
Yuqing Kong and Grant Schoenebeck.
\newblock An information theoretic framework for designing information
  elicitation mechanisms that reward truth-telling.
\newblock \emph{ACM Trans. Econ. Comput.}, 7\penalty0 (1):\penalty0 2:1--2:33,
  January 2019.
\newblock ISSN 2167-8375.
\newblock \doi{10.1145/3296670}.
\newblock URL \url{http://doi.acm.org/10.1145/3296670}.

\bibitem[Liu and Chen(2018)]{liuchen}
Yang Liu and Yiling Chen.
\newblock Surrogate scoring rules and a dominant truth serum for information
  elicitation.
\newblock \emph{CoRR}, abs/1802.09158, 2018.
\newblock URL \url{http://arxiv.org/abs/1802.09158}.

\bibitem[Miller et~al.(2005)Miller, Resnick, and Zeckhauser]{MRZ05}
N.~Miller, P.~Resnick, and R.~Zeckhauser.
\newblock Eliciting informative feedback: The peer-prediction method.
\newblock \emph{Management Science}, pages 1359--1373, 2005.

\bibitem[Prelec(2004)]{prelec2004bayesian}
D.~Prelec.
\newblock A {B}ayesian {T}ruth {S}erum for subjective data.
\newblock \emph{Science}, 306\penalty0 (5695):\penalty0 462--466, 2004.

\bibitem[Prelec et~al.(2017)Prelec, Seung, and McCoy]{prelec2017solution}
Dra{\v{z}}en Prelec, H~Sebastian Seung, and John McCoy.
\newblock A solution to the single-question crowd wisdom problem.
\newblock \emph{Nature}, 541\penalty0 (7638):\penalty0 532--535, 2017.

\bibitem[Radanovic and Faltings(2014)]{radanovic2014incentives}
Goran Radanovic and Boi Faltings.
\newblock Incentives for truthful information elicitation of continuous
  signals.
\newblock In \emph{Twenty-Eighth AAAI Conference on Artificial Intelligence},
  2014.

\bibitem[Schoenebeck and Yu(2019)]{SchoenebeckY2019}
Grant Schoenebeck and Fang-Yi Yu.
\newblock Robust and strongly truthful multi-task peer prediction mechanisms
  for heterogeneous agents.
\newblock Unpublished manuscript, 2019.

\bibitem[Seneta(2006)]{seneta2006non}
Eugene Seneta.
\newblock \emph{Non-negative matrices and Markov chains}.
\newblock Springer Science \& Business Media, 2006.

\bibitem[Shannon(2001)]{Shannon:2001:MTC:584091.584093}
C.~E. Shannon.
\newblock A mathematical theory of communication.
\newblock \emph{SIGMOBILE Mob. Comput. Commun. Rev.}, 5\penalty0 (1):\penalty0
  3--55, January 2001.
\newblock ISSN 1559-1662.
\newblock \doi{10.1145/584091.584093}.
\newblock URL \url{http://doi.acm.org/10.1145/584091.584093}.

\bibitem[Shnayder et~al.(2016)Shnayder, Agarwal, Frongillo, and
  Parkes]{2016arXiv160303151S}
Victor Shnayder, Arpit Agarwal, Rafael Frongillo, and David~C Parkes.
\newblock Informed truthfulness in multi-task peer prediction.
\newblock In \emph{Proceedings of the 2016 ACM Conference on Economics and
  Computation}, pages 179--196. ACM, 2016.

\bibitem[Witkowski and Parkes(2012)]{witkowski2012robust}
J.~Witkowski and D.~Parkes.
\newblock A robust {B}ayesian {T}ruth {S}erum for small populations.
\newblock In \emph{Proceedings of the 26th AAAI Conference on Artificial
  Intelligence (AAAI 2012)}, 2012.

\end{thebibliography}

\appendix

\section{Additional proof}\label{appendix:a}

\begin{fact}\label{fact:dmi}
    For every square transition matrix $\mathbf{T}$, $|\det(\mathbf{T})|\leq 1$ and the equality holds if and only if $\mathbf{T}$ is a permutation matrix. 
\end{fact}

The above fact can be implied by Perron-Frobenius theorem \cite{seneta2006non} and here we present a direct and basic proof. 

\begin{proof}
    Since a matrix's determinant equals its transpose's determinant, without loss of generality, we assume $\mathbf{T}$ is row-stochastic, i.e., every row sums to 1. A matrix $\mathbf{A}$ is sub-row-stochastic if its all elements are in $[0,1]$ and every row's sum is less than 1. We will show a slightly stronger result: for every square sub-row-stochastic matrix $\mathbf{A}$, $|\det(\mathbf{A})|\leq 1$ and the equality holds iff $\mathbf{A}$ is a permutation matrix. This result implies the fact and is easy to be proved by induction with respect to the matrix dimension $C$. 
    
    When $C=1$, the result holds naturally. For $C>1$, when the result holds for $C-1$, first note that
    \begin{align*}
        |\det(\mathbf{A})|= & |\sum_j (-1)^{j+1} a_{1,j}\det(\mathbf{A}_{-1,-j})|\\
        \leq & \sum_j a_{1,j}|\det(\mathbf{A}_{-1,-j})|
    \end{align*}
    
    where $a_{1,j}$ is the $j^{th}$ element of $\mathbf{A}$'s first row and $\mathbf{A}_{-1,-j}$ is $\mathbf{A}$'s minor that deletes the first row and the $j^{th}$ column. 

Since every minor is also sub-row-stochastic and the result holds for $C-1$, we have
\begin{align*}
        |\det(\mathbf{A})|\leq \sum_j a_{1,j}|\det(\mathbf{A}_{-1,-j})|\leq \sum_j a_{1,j}\leq 1
    \end{align*}
    
When $|\det(\mathbf{A})|=1$, $\sum_i a_{1,j}=1$ and for all $j'$ that $a_{1,j'}>0$, $|\det(\mathbf{A}_{-1,-j'})|=1$, which implies that $\mathbf{A}_{-1,-j'}$ is a permutation matrix (recall that the result holds for $C-1$). We pick $j_0$ such that $a_{1,j_0}>0$. Since $\mathbf{A}_{-1,-{j_0}}$ is a permutation matrix, due to the fact that $\mathbf{A}$ is sub-row-stochastic, for all $i\neq 1$, $a_{i,j_0}$ must be zero. Then for all $j\neq j_0$, $\det(\mathbf{A}_{-1,-{j}})$ must be zero since it has all zero column. Thus $a_{1,j_0}$ must be one to make $|\det(\mathbf{A})|=1$. The fact that $\mathbf{A}$ is sub-row-stochastic implies that $a_{1,j}=0, \forall j\neq j_0$. Therefore, $\mathbf{A}$ must be a permutation matrix to make $|\det(\mathbf{A})|=1$.      
\end{proof}

\end{document}